\documentclass[submission,copyright,creativecommons]{eptcs}
\providecommand{\event}{LTT 2026} 

\usepackage{iftex}

\ifpdf
  \usepackage{underscore}         
  \usepackage[T1]{fontenc}        
\else
  \usepackage{breakurl}           
\fi

\usepackage[T1]{fontenc}

\usepackage[english]{babel}

\usepackage{graphicx}
\usepackage{amsmath}
\usepackage{amssymb}
\usepackage{bm}
\usepackage{prooftree}
\usepackage{hyperref}

\newcommand{\Lang}[1]{{\mathcal L}(#1)}

\newcommand{\emptys}{\langle \, \rangle}

\newcommand{\typeof}[1]{\pvec \tau_{#1}}

\newcommand{\iInter}[3]{#2 \,\mathop{ur}_{#3} #1}

\newcommand{\herbrandInter}[3]{#2 \,\mathop{hr}_{#3} #1}
\newcommand{\kleeneInter}[3]{#2 \,\mathop{r}_{#3} #1}

\newcommand{\kripkeInter}[3]{#2 \,\Vdash_{#3} #1}

\newcommand{\kreiselInter}[3]{#2 \,\mathop{mr}_{#3} #1}
\newcommand{\bergerInter}[2]{#2 \,\mathop{mr}\!{}_{\bot} \, #1}

\newcommand{\pforall}[2]{\forall^{#2} {#1}}
\newcommand{\pexists}[2]{\exists^{#2} {#1}}

\newcommand{\pvec}[1]{\boldsymbol{#1}}
\newcommand{\proves}{\vdash}

\newcommand{\std}{{\sf st}}

\newcommand{\K}{\mathcal{K}_1}
\newcommand{\T}{\mathcal{T}}

\newcommand{\bounded}[4]{#2 \lhd^{#4}_{#1} #3}

\newcommand{\SourceT}{\mathbf{S}}
\newcommand{\TargetT}{\mathbf{T}}

\newcommand{\HA}{\mathbf{HA}}

\newcommand{\NN}{\mathbb{N}}

\newcommand{\nat}{\mathbf{nat}}

\newtheorem{definition}{Definition}[section]
\newtheorem{proposition}[definition]{Proposition}

\newtheorem{theorem}[definition]{Theorem}
\newtheorem{remark}[definition]{Remark}

\title{Uniform Realizability Interpretations}

\author{Ulrich Berger
\institute{Swansea University}
\and Paulo Oliva
\institute{Queen Mary University of London}}

\def\titlerunning{Uniform Realizability Interpretations}
\def\authorrunning{U. Berger \& P. Oliva}

\begin{document}
%
%
%

\maketitle              
\begin{abstract}
This work introduces a novel framework of \emph{uniform realizability} that unifies and generalizes various realizability interpretations of logic, particularly focussing on the treatment of atomic formulas and quantifiers. Traditional realizability interpretations (such as Kleene's number realizability) require explicit witnesses for existential quantifiers. In contrast, newer approaches, such as in the first author's uniform Heyting arithmetic, Herbrand realizability of non-standard arithmetic, or in the ``classical'' realizability of arithmetic, (some) quantifiers, are treated uniformly. The proposed notion of uniform realizability abstracts these differences, parametrising the interpretation by a given treatment of atomic formulas, accounting for both classical and modern variants. The approach is illustrated using several realizability interpretations of Heyting arithmetic.
\end{abstract}


\quad \emph{Dedicated to Stefano Berardi on his 64th Birthday}

\section{Introduction}

Since Kleene's \emph{realizability interpretation} of Heyting arithmetic \cite{Kleene(45)} in 1945 several variants of realizability have been proposed. These include Kreisel's modified realizability \cite{Kreisel(59),Kreisel(62)}, `classical' realizability \cite{BO(02B),BO(02A)}, van den Berg et al. `Herbrand' realizability of non-standard arithmetic \cite{Berg(2012)}, and the Aschieri-Berardi `learning' realizability \cite{Aschieri(2010)}. 

In most of these variants, a realizer of a formula carries enough information to fully or partially witness the formula's positive existential quantifiers. Realizers also (normally) take as input the universally quantified data. For instance, in Kleene's realizability by natural numbers, the clauses for $\exists n A(n)$ and $\forall n A(n)$ are as follows:
\begin{align*}
	\kleeneInter{\,\exists n A(n)}{k}{} 
		& \;\; :\equiv \;
		\kleeneInter{A(k_0)}{k_1}{} \\
	\kleeneInter{\,\forall n A(n)}{k}{} 
		& \;\; :\equiv \;
		\forall n (\{k\}(n) \downarrow \wedge \, \kleeneInter{A(n)}{\{k\}(n)}{}).
\end{align*}
Therefore, a formula $\exists n A(n)$ is realized by a number $k$, if, when viewing $k$ as a pair $\langle k_0, k_1 \rangle$, we have that $k_0$ witnesses $n$, and $k_1$ is a realizer for $A(k_0)$. In the case of the universal quantifier, we view $k$ as a function that, for each $n$, produces a realiser for $A(n)$.

However, in some cases, quantifiers are treated \emph{uniformly}, i.e. a realiser for $\exists n A(n)$ does not explicitly provide information on $n$, and a realizer for $\forall n A(n)$ must be a realizer for $A(n)$ uniformly in $n$ (the same realiser must work for all $n$). For instance, the interpretation of \emph{internal} quantifiers in the Herbrand realizability of non-standard arithmetic \cite{Berg(2012)} is as follows:
\begin{align*}
	\herbrandInter{\exists n A(n)}{\pvec a}{} 
		& \;\; :\equiv \;
		\exists n (\herbrandInter{A(n)}{\pvec a}{}) \\
	\herbrandInter{\forall n A(n)}{\pvec a}{} 
		& \;\; :\equiv \;
		\forall n (\herbrandInter{A(n)}{\pvec a}{}).
\end{align*}
This is also the case for the realizability interpretation of the \emph{non-computational} quantifiers in Uniform Heyting arithmetic \cite{Berger(05c)}.

The situation is even more interesting in other, more recent variants of realizability, such as the Aschieri-Berardi learning realizability \cite{Aschieri(2010)}, where the interpretation is parametrized by a `state of knowledge' $s$, and realizers simply provide a one-time update on the state -- proofs are viewed as state updating transformers.

We will introduce a notion of \emph{uniform realizability} capturing the idea of uniform quantifiers. 
The new notion is parametric in the space of (potential) realizers and the interpretation of atomic predicates,
and contains the examples above as instances,
thus isolating the aspects where they agree and where they differ.

In the following, we make a distinction between unqualified (or untyped or uniform) quantifiers, i.e.
\begin{equation} \label{quant-untyped}
	\exists x A(x) \quad \quad \mbox{and} \quad \quad \forall x A(x),
\end{equation}
and qualified (or typed) quantifiers
\begin{equation} \label{quant-typed}
	\pexists{x}{B} A(x) \quad \quad \mbox{and} \quad \quad \pforall{x}{B} A(x),
\end{equation}
where $B$ is a predicate.

We propose that unqualified quantifiers of the form (\ref{quant-untyped}) should always be given a \emph{uniform realizability interpretation}, while qualified existential and universal quantifications (\ref{quant-typed}) should be treated as abbreviations for 
\begin{equation} \label{quant-typed-abbreviation}
	\exists x(B(x) \wedge A(x)) \quad \quad \mbox{and} \quad \quad \forall x(B(x) \to A(x)),
\end{equation}
respectively, so that their interpretation follows from the (uniform) interpretation of the unqualified quantifiers together with an interpretation for $B(x)$. 
%
In particular, in the case of Heyting arithmetic $\HA$, we write $\pexists{n}{\NN} A(n)$ and $\pforall{n}{\NN} A(n)$ as abbreviations for 
\begin{equation} \label{quant-typed-nat}
	\exists n (\NN(n) \wedge A(n)) \quad \quad \mbox{and} \quad \quad \forall n (\NN(n) \to A(n)),
\end{equation}
so that we only need to focus on the interpretation of the predicate $\NN(n)$.

The uniform realizability interpretation presented here complements the recent work by the second author on a uniform functional interpretation (Diller-Nahm variant) of first-order logic \cite{Oliva(2025)}.

\subsection{Heyting arithmetic $\HA$}
\label{sec-HA}

\newcommand{\suc}{{\sf succ}}

Consider Heyting (intuitionistic) arithmetic $\HA$ formulated with three predicate symbols: falsity $\perp$ (nullary), natural number $\NN$ (unary), and equality $=$ (binary). This means that in $\HA$ we have three kinds of atomic formulas:
\begin{align}
	\perp & \qquad \mbox{falsity} \\
	\NN(n) & \qquad \mbox{$n$ is a number} \\
	n = m & \qquad \mbox{equality}.
\end{align}
We will also assume two function symbols: zero $0$ (nullary) and successor $\suc$ (unary). So, the terms of the language are the numerals: $0, \suc(0), \ldots$, normally denoted by $\overline{n}$ for $\suc^{(n)}(0)$.

\begin{table}[ht] 
\[
    \begin{array}{|rccc|}
        \hline
        & & & \\
        \multicolumn{4}{|c|}{ 
        \begin{prooftree} 
            \justifies 
            \Gamma, A \proves A
            \using (\mbox{logical axiom})
        \end{prooftree}} \\[4mm]
        & 
        \begin{prooftree}
            \Gamma \proves A
            \quad
            \Gamma \proves B
            \justifies 
            \Gamma \proves A \wedge B
            \using (\wedge\mbox{-intro})
        \end{prooftree}
        & &
        \begin{prooftree} 
            \Gamma \proves A_1 \wedge A_2
            \justifies
            \Gamma \proves A_i
            \using (\wedge\mbox{-elim}_i)
        \end{prooftree} \\[5mm]
        & 
        \begin{prooftree}
            \Gamma, A \proves B
            \justifies 
            \Gamma \proves A \to B
            \using (\to\mbox{-intro})
        \end{prooftree}
        & &
        \begin{prooftree} 
            \Gamma \proves A \to B
            \quad 
            \Gamma \proves A
            \justifies
            \Gamma \proves B
            \using (\to\mbox{-elim})
        \end{prooftree} \\[5mm]
        & 
        \begin{prooftree}
            \Gamma \proves A(s)
            \justifies 
            \Gamma \proves \exists n A(n)
            \using (\exists\mbox{-intro})
        \end{prooftree}
        & & 
        \begin{prooftree}
            \Gamma \proves \exists n A(n)
            \quad
            \Gamma, A(n) \proves B
            \justifies
            \Gamma \proves B
            \using (\exists\mbox{-elim})
        \end{prooftree} \\[5mm]
        & 
        \begin{prooftree}
            \Gamma \proves A(n)
            \justifies 
            \Gamma \proves \forall n A(n)
            \using (\forall\mbox{-intro})
        \end{prooftree}
        & & 
        \begin{prooftree}
            \Gamma \proves \forall n A(n)
            \justifies
            \Gamma \proves A(s)
            \using (\forall\mbox{-elim})
        \end{prooftree} \\[6mm]
        \hline
    \end{array}
\]
\caption{Logical axioms and rules of $\HA$} 
\label{HA-logical-rules}
\end{table}

The logical axioms and rules are shown in Figure \ref{HA-logical-rules} (with the usual side-conditions on the quantifier rules) and they correspond to the minimal logic fragment of $\HA$. We consider the following as the \emph{non-logical axioms} of the theory $\HA$. For the predicate $\perp$ we assume the ex-falso-quodlibet axiom schema:
\begin{equation}
	\perp \; \proves A.
\end{equation}
Equality and $\NN(n)$ are assumed to satisfy:
\begin{align}
	& \forall n (n = n) & \mbox{(reflexivity)} \\
	& \forall n, m (n = m \to m = n) & \mbox{(symmetry)} \\
	& \forall n, i, m (n = i \wedge i = m \to n = m) & \mbox{(transitivity)} \\
	& \forall n, m (\NN(n) \wedge n = m \to \NN(m)) & \mbox{($\NN$ closed under equality)}
\end{align}
The arithmetic (Peano) axioms are as follows:
\begin{align}
	& \NN(0) & \mbox{($0$ is in $\NN$)} \\
	& \forall n (\NN(n) \to \NN(\suc(n)) & \mbox{($\NN$ closed under successor)} \\
	& \pforall{n, m}{\NN} (\suc(n) = \suc(m) \to n = m) & \mbox{(successor  injective on $\NN$)} \\
	& \pforall{n}{\NN} (\suc(n) \neq 0) & \mbox{($0$ not successor of $n$ in $\NN$)}
\end{align}
where $\pforall{n}{\NN} A$ abbreviates $\forall n (\NN(n) \to A)$. 

There are two options for the induction scheme, depending on
whether the induction step is restricted to $\NN$ or not:
\begin{equation} \label{ax-induction}
	A(0) \wedge \pforall{n}{\NN} (A(n) \to A(\suc(n)) \to \pforall{n}{\NN} A(n).
\end{equation}
\begin{equation} \label{ax-induction-u}
	A(0) \wedge \forall n (A(n) \to A(\suc(n)) \to \pforall{n}{\NN} A(n).
\end{equation}
The latter is seemingly weaker than the former, but (\ref{ax-induction}) can be obtained from (\ref{ax-induction-u}) by using (\ref{ax-induction-u}) with $A(n) \wedge \NN(n)$. If $\NN(n)$ is interpreted in the standard way, then (\ref{ax-induction}) is realized by the usual primitive recursion operator, while the realizer of (\ref{ax-induction-u}) is a pure iterator, without access to the recursion argument. (\ref{ax-induction-u}) expresses that $\NN$ is the least set containing $0$ and being closed under successor, which is the induction scheme naturally obtained in a formal treatment of general strictly positive inductive definitions (see for example~\cite{BergerTsuiki21}). In Section~\ref{sec-nat} we discuss realizability of induction in case $\NN(n)$ is not interpreted in the standard way.


\section{The Uniform Realizability Interpretation}
\label{sec-uniform}

In the following we will present a general (parametrized) realizability interpretation of an arbitrary \emph{source theory} $\SourceT$. We will later consider the case where $\SourceT$ is the theory $\HA$ or some variant thereof. The realizability interpretation of $\SourceT$ will be into some \emph{target theory} $\TargetT$. We will not give formal conditions on what $\TargetT$ should be, other than that it is closed under the rules of intuitionistic logic and that it has a sort for (potential) realizers, with function application $f(\pvec a)$ -- where $\pvec a$ is a tuple of realizers and the result is possibly undefined -- and $\lambda$-abstraction, i.e. for any term $t[\pvec a]$ there is a term $\lambda \pvec a . t[\pvec a]$ such that $\TargetT$ proves
\begin{equation}
	(\lambda \pvec a . t[\pvec a])(\pvec s) = t[\pvec s]
\end{equation}
%
%
We write $a{\downarrow}$ to denote that $a$ does not diverge. 
We leave it open whether these $\lambda$-terms are typed or untyped, as we will consider both cases. Besides a sort for realizers and a sort for natural numbers (which may or may not be the same) $\TargetT$ may have other sorts and types which, however, we usually leave implicit, unless they improve readability (as, for example, for modified- and Herbrand-realizability in Sections~\ref{sec-kreisel} and~\ref{sec-herbrand}).

\begin{definition}[Base interpretation of $\Lang{\SourceT}$ into $\Lang{\TargetT}$] \label{def-base} A \emph{base interpretation} of $\Lang{\SourceT}$ into $\Lang{\TargetT}$ associates to each $n$-ary predicate symbol $P$ of the language of $\SourceT$ an $(n+m)$-ary relation $\bounded{P}{\pvec x}{\pvec a}{}$ in the language of $\TargetT$, between tuples $\pvec x$ (arity $n$) and $\pvec a$ (arity $m$, for some $m$). We read this as \emph{$\pvec x$ is $P$-bounded by $\pvec a$}.
\end{definition}

We think of the tuple $\pvec a$ as the realizers or witness of $P(\pvec x)$. Either $\pvec x$ or $\pvec a$ could be the empty (nullary) tuple. Whenever we need to explicitly write the empty tuple we will use the symbol $\emptys$. \\[2mm]
{\bf Notation}. If $\pvec f$ is an $m$-tuple, $f_1, \ldots, f_m$, we write $\pvec f(\pvec a)$ for the $m$-tuple $f_1(\pvec a), \ldots, f_m(\pvec a)$, and ${\pvec f}(\pvec a){\downarrow}$ for the conjunction $f_1(\pvec a){\downarrow}\wedge\ldots\wedge f_m(\pvec a){\downarrow}$.

\begin{definition}[Uniform realizability interpretation] \label{def-uniform-realizability} Let a base interpretation of $\Lang{\SourceT}$ into $\Lang{\TargetT}$ be given. For each formula $A$ of $\SourceT$, possibly with free-variables, associate a formula ($\pvec a$ uniformly realises $A$)
\begin{equation}
	\iInter{A}{\pvec a}{}
\end{equation}	
of $\TargetT$, by induction on $A$. For atomic formulas $P(\pvec x)$ the interpretation is as in the base interpretation:
\begin{align}
	\iInter{P(\pvec x)}{\pvec a}{} 
		& \;\; :\equiv \; \bounded{P}{\pvec x}{\pvec a}{}.
\end{align}
So, $\pvec a$ uniformly realizes $P(\pvec x)$ if $\pvec x$ is $P$-bounded by $\pvec a$. 
For composite formulas the interpretation is defined as follows:
\begin{align}
	\iInter{A \wedge B}{\pvec a, \pvec b}{} 
		& \;\; :\equiv \;
		(\iInter{A}{\pvec a}{}) \wedge (\iInter{B}{\pvec b}{}) \\
	\iInter{A \to B}{\pvec f}{} 
		& \;\; :\equiv \;
		\forall \pvec a ((\iInter{A}{\pvec a}{}) \to (\pvec f(\pvec a) {\downarrow}) \wedge (\iInter{B}{\pvec f(\pvec a)}{})) \\
	\iInter{\exists x A(x)}{\pvec a}{} 
		& \;\; :\equiv \;
		\exists x (\iInter{A(x)}{\pvec a}{}) \\
	\iInter{\forall x A(x)}{\pvec a}{} 
		& \;\; :\equiv \;
		\forall x (\iInter{A(x)}{\pvec a}{}).
\end{align}
%
By induction on formulas one easily sees that realizability commutes with substitution, i.e.
\begin{equation} \label{eq-substitution}
   \iInter{\,(A[t/x])}{\pvec a}{} \; 
   \equiv \; (\iInter{A}{\pvec a}{})[t/x]
\end{equation}
provided $x$ does not occur in $\pvec a$.
\end{definition}

\begin{remark}[Total realizers] 
%
%
In case realizers are \emph{total} ($a{\downarrow}$ always holds), 
the clause for implication simplifies to
\begin{align}
	\iInter{A \to B}{\pvec f}{} 
		& \;\; :\equiv \;
		\forall \pvec a ((\iInter{A}{\pvec a}{}) \to (\iInter{B}{\pvec f(\pvec a)}{})). 
\end{align}
\end{remark}

\begin{remark}[Realizing qualified quantifiers] \label{rem-bounded-quantifers} 
Recall that for any basic predicate $P(\cdot)$ in the source theory $\SourceT$ we introduced the qualified quantifiers
\begin{align}
\pforall{\pvec x}{P} A({\pvec x}) 
       & \;\; :\equiv
    \forall {\pvec x}\,(P({\pvec x}) \to A({\pvec x}))\\
\pexists{\pvec x}{P} A({\pvec x}) 
       & \;\; :\equiv
    \exists {\pvec x}\,(P({\pvec x}) \wedge A({\pvec x}))
\end{align}
If we introduce the following abbreviations in the target theory $\TargetT$
\begin{align}
\forall \bounded{P}{\pvec x}{\pvec b}{}\, \phi({\pvec b},{\pvec x})
       & \;\; :\equiv
    \forall {\pvec x}\,(\bounded{P}{\pvec x}{\pvec b}{} \to \phi({\pvec b},{\pvec x}))\\
\exists \bounded{P}{\pvec x}{\pvec b}{} \,\phi({\pvec b},{\vec x})
       & \;\; :\equiv
    \exists {\pvec x}\,(\bounded{P}{\pvec x}{\pvec b}{} \land \phi({\pvec b},{\pvec x})),
\end{align}
%
then we have
\begin{align}
\iInter{\pforall{\pvec x}{P} A({\pvec x})}{{\pvec f}}{}
       & \;\; \Leftrightarrow \;\;
\forall{\pvec b}\, \forall\bounded{P}{\pvec x}{\pvec b}{}\,({\pvec f}({\pvec b}){\downarrow}\wedge\iInter{A({\pvec x})}{{\pvec f}({\pvec b})}{})\\
\iInter{\pexists{\pvec x}{P} A({\pvec x})}{{\pvec b},{\pvec a}}{}
       & \;\; \Leftrightarrow \;\;
\exists\bounded{P}{\pvec x}{\pvec b}{}\,(\iInter{A({\pvec x})}{{\pvec a}}{}).
\end{align}
In the special case that $\bounded{P}{}{}{}$ is equality, i.e.\ $(\bounded{P}{\pvec x}{\pvec b}{}) \Leftrightarrow (\pvec x = \pvec b)$, this simplifies to (using (\ref{eq-substitution}))
\begin{align}
\iInter{\pforall{\pvec x}{P} A({\pvec x})}{{\pvec f}}{}
       & \;\; \Leftrightarrow \;\;
\forall{\pvec b}\,({\pvec f}({\pvec b}){\downarrow}\wedge\iInter{A({\pvec b})}{{\pvec f}({\pvec b})}{})\\
\iInter{\pexists{\pvec x}{P} A({\pvec x})}{{\pvec b},{\pvec a}}{}
       & \;\; \Leftrightarrow \;\;
\iInter{A({\pvec b})}{{\pvec a}}{}.
\end{align}
\end{remark}

\begin{remark}[Uniform predicates and formulas] \label{rem-uniform} Let us call a predicate $P$ \emph{uniformly interpreted}, or just \emph{uniform}, if $\bounded{P}{\pvec x}{\emptys}{} \equiv P(\pvec x)$, and call a formula $A$ \emph{uniform} if it contains only uniform predicates. Then it is easy to see that if $A$ is uniform, then the formula $\iInter{A}{\emptys}{}$ is syntactically identical to $A$ (${\pvec f}({\pvec b}){\downarrow}$ can be omitted since $\pvec f$ will be the empty tuple).
\end{remark}

\begin{definition}[Realizable sequents and formulas] For a fixed base interpretation of $\Lang{\SourceT}$, we say that a sequent $\Gamma \proves A$ of $\SourceT$ is \emph{realizable} if for some $\lambda$-term $\pvec t[\pvec \gamma]$ of $\TargetT$, with $\pvec \gamma$ as the only free-variables, we have
\begin{equation}
	(\pvec \gamma \downarrow), (\iInter{\Gamma}{\pvec \gamma}{}) \proves_{\TargetT} (\pvec t[\pvec \gamma] \downarrow) \wedge (\iInter{A}{\pvec t[\pvec \gamma]}{}).
\end{equation}
A closed formula $A$ is realizable if the sequent $\proves_{\TargetT} A$ is realizable.
\end{definition}

\begin{theorem}[Soundness] 
\label{thm-soundness}
Given a base interpretation of $\Lang{\SourceT}$, if all the non-logical axioms of $\SourceT$ are realizable then all the theorems of $\SourceT$ are realizable.
\end{theorem}
{\bf Proof}. 
By induction on derivations one shows that all derivable sequents are realizable.
The non-logical axioms are realizable by assumption. Since the quantifiers are treated uniformly, their interpretation is straightforward -- making use of observation (\ref{eq-substitution}). For instance: \\[1mm]
\emph{$\exists$-introduction}. Assuming
\[ 
	(\pvec \gamma \downarrow), (\iInter{\Gamma}{\pvec \gamma}{}) \proves (\pvec t[\pvec \gamma] \downarrow) \wedge (\iInter{A(s)}{\pvec t[\pvec \gamma]}{}),
\]
since $(\iInter{A(s)}{\pvec t[\pvec \gamma]}{}) \equiv (\iInter{A(x)}{\pvec t[\pvec \gamma]}{})[s/x]$, we have
\[ 
	(\pvec \gamma \downarrow), (\iInter{\Gamma}{\pvec \gamma}{}) \proves (\pvec t[\pvec \gamma] \downarrow) \wedge \exists x (\iInter{A(x)}{\pvec t[\pvec \gamma]}{})
\]
and hence
\[ 
	(\pvec \gamma \downarrow), (\iInter{\Gamma}{\pvec \gamma}{}) \proves (\pvec t[\pvec \gamma] \downarrow) \wedge (\iInter{\exists x A(x)}{\pvec t[\pvec \gamma]}{}).
\]
\emph{Logical axiom}. We can take $\pvec t[\pvec a] = \pvec a$ since we have:
\[
	(\pvec a \downarrow) , (\iInter{A}{\pvec a}{}) \proves (\pvec a \downarrow) \wedge (\iInter{A}{\pvec a}{}).
\]
\emph{$\to$-intro}. Assuming
\[ 
	(\pvec \gamma \downarrow), (\iInter{\Gamma}{\pvec \gamma}{}), (\iInter{A}{\pvec a}{}) \proves (\pvec t[\pvec \gamma, \pvec a] \downarrow) \wedge (\iInter{B}{\pvec t[\pvec \gamma, \pvec a]}{}).
\]
we have
\[ 
	(\pvec \gamma \downarrow), (\iInter{\Gamma}{\pvec \gamma}{}) \proves (\lambda \pvec a . \pvec t[\pvec \gamma, \pvec a] \downarrow) \wedge (\iInter{A \to B}{\lambda \pvec a . \pvec t[\pvec \gamma, \pvec a]}{})
\]
using $(\lambda \pvec a . \pvec t[\pvec \gamma, \pvec a])(\pvec a) = \pvec t[\pvec \gamma, \pvec a]$. \\[2mm]
\emph{$\to$-elim}. Assume
\[ 
	(\pvec \gamma \downarrow), (\iInter{\Gamma}{\pvec \gamma}{}) 
		\proves 
		(\pvec t [\pvec \gamma] \downarrow) \wedge 
			\forall \pvec a (\iInter{A}{\pvec a}{} \to \pvec t[\pvec \gamma](\pvec a) \downarrow 
				\wedge \, (\iInter{B}{\pvec t[\pvec \gamma](\pvec a)}{}))
\]
and
\[ 
	(\pvec \gamma \downarrow), (\iInter{\Gamma}{\pvec \gamma}{})
		\proves \pvec s[\pvec \gamma] \downarrow \wedge \, (\iInter{A}{\pvec s[\pvec \gamma]}{}).
\]
Then
\[ 
	(\pvec \gamma \downarrow), (\iInter{\Gamma}{\pvec \gamma}{})
		\proves 
		\pvec t[\pvec \gamma] (\pvec s[\pvec \gamma]) \downarrow 
			\wedge \, (\iInter{B}{\pvec t[\pvec \gamma] (\pvec s[\pvec \gamma])}{}).
\]
The other cases are treated similarly. \hfill $\Box$ \\

The definition \ref{def-uniform-realizability} describes how we can extend a given \emph{base} interpretation to a \emph{full} interpretation. In the following section, we will look at particular choices of base interpretation for $\HA$ and show that the full interpretations obtained coincide with (or are very close to) various well-known realizability interpretations of $\HA$. First, however, let us consider some general (abstract) base interpretations of $n = m$, $\bot$ and $\NN(n)$.

\subsection{Interpreting equality $n = m$}
\label{sec-equality}

Let us first consider general interpretations of the (binary) equality predicate $n = m$. In most of our instances we will assume that $n = m$ is interpreted uniformly. For the Aschieri-Berardi learning realizability, however, a non-trivial interpretation of equality of the form 
\begin{equation}
	(n, m) \triangleleft_{=} a \quad :\equiv \quad \phi(a) \to n = m,
\end{equation}
for some formula $\phi(a)$, is used.
Then, for the equality axioms, symmetry and transitivity, to be realizable, one needs terms $s$ and $t$ such that
\begin{align}
	\phi(s(a)) \proves_{\TargetT} \; 
		& \phi(a) & \mbox{(symmetry)} \\
	\phi(t(a_1, a_2)) \proves_{\TargetT} \; 
		& \phi(a_1) \wedge \phi(a_2) & \mbox{(transitivity)}
\end{align}
while the axiom of reflexivity follows from reflexivity itself by weakening. More on this in Section \ref{sec-learning} (see proof of Proposition \ref{prop-berardi}).

\subsection{Interpreting falsity $\bot$}
\label{sec-falsity}

As with the equality predicate discussed above, in most of our instances we will assume that falsity $\bot$ is interpreted uniformly. However, for the classical (Section \ref{sec-classical}) and learning (Section \ref{sec-learning}) realizability, a non-trivial interpretation of $\bot$ is assumed. Consider a reasonably general case where $\bot$ is witnessed by a unary predicate $\psi(a)$:
\begin{equation}
	\emptys \triangleleft_{\bot} a \quad :\equiv \quad \psi(a).
\end{equation}
Consider the ex-falso-quodlibet axiom schema:
\[
	\perp \; \proves A.
\]
We would need to have terms $\pvec t[\pvec a]$ (depending on $A$) such that 
\begin{equation}
	\psi(a) \proves_{\TargetT} \iInter{A}{\pvec t[a]}{}.
\end{equation}
Friedman's \cite{Friedman(78)} and Dragalin's \cite{Dragalin(80)} idea is to also interpret all atomic formulas (including equality) using the same unary predicate, e.g.
\begin{equation}
	(n, m) \triangleleft_{=} a \quad :\equiv \quad (n = m) \vee \psi(a).
\end{equation}
In this way, one can easily define $\pvec t[\pvec a]$ for every formula $A$.

\subsection{Interpreting $\NN(n)$} 
\label{sec-nat}

In the following, we will consider different interpretations $n \triangleleft_{\NN} a$ of the atomic formula $\NN(n)$. They can be classified as follows:
\begin{itemize}
    \item \emph{Uniform interpretations}: In this case $a$ is the empty tuple and $\NN(n)$ is always realizable:
    \[ n \triangleleft_{\NN} \emptys \quad :\equiv \quad \mathrm{true}, \]
    which corresponds to treating $\pforall{n}{\NN} A(n)$ and $\pexists{n}{\NN} A(n)$ simply as $\forall n A(n)$ and $\exists n A(n)$ (uniform quantifications).
    This will be used in Section \ref{sec-herbrand} (Herbrand realizability) to give \emph{internal quantifiers} a uniform interpretation.
    \item \emph{Approximate interpretations}: In this case $a$ provides some partial information about $n$. For instance, $a$ could be a finite set of numbers that contains $n$
    \[ n \triangleleft_{\NN} a \quad :\equiv \quad n \in a. \]
    An example of this is Lifschitz' realizability~\cite{vanOosten90}.
    Also, quantification over standard natural numbers $\std(n)$, so-called \emph{external quantifiers} in the Herbrand realizability (Section \ref{sec-herbrand}) follows this approach by interpreting $\std(n)$ as
    \[ n \triangleleft_{\std} a \quad :\equiv \quad n \in a. \]
    Another example of an approximate interpretation would be to treat the realizer for $\NN(n)$ as an upper bound on $n$:
    \[ n \triangleleft_{\NN} a \quad :\equiv \quad n \leq a. \]
    This is used in bounded modified realizability interpretation~\cite{Ferreira(05A)}.
    \item \emph{Precise interpretation}: In this case $a$ is exactly $n$. Therefore, $n$ being a natural number is witnessed by a number $a$ such that $n = a$. Formally:
    \[ n \triangleleft_{\NN} a \quad :\equiv \quad n = a. \]
    Examples of this are the Kleene and Kreisel realizability interpretations (Sections \ref{sec-kleene} and \ref{sec-kreisel}).
\end{itemize}
In all cases one must choose the interpretation in such a way that all non-logical axioms involving the predicate $\NN(n)$ can be realized. For instance, let us consider the interpretation of the induction schema if $n \triangleleft_{\NN} m \; :\equiv \; n \leq m$: Given a realiser $\pvec t$ for $A(0)$, i.e.,
\begin{equation} \label{ind-base}
    \iInter{A(0)}{\pvec t}{},    
\end{equation}
and a realizer $\pvec \phi$ for the induction step,
\begin{equation} \label{ind-step}
    \forall m \forall n \leq m \forall \pvec a ((\iInter{A(n)}{\pvec a}{}) \to (\iInter{A(n + 1)}{\pvec \phi(m, \pvec a)}{})),
\end{equation}
we have to compute a realizer $\pvec \psi$ of $\pforall{n}{\NN} A(n)$, i.e.
\begin{equation}\label{ind-conclusion}
\forall m\forall n \le m (\iInter{A(n)}{\pvec \psi(m)}{}).
\end{equation}
Notice that the realizer $\pvec \psi(m)$ must work uniformly for all $n \leq m$.
This suggests to assume a partial order w.r.t. which realizers are upwards closed
\begin{equation} \label{upper-bound}
   {\pvec a} \preceq {\pvec b}\ \to\ \iInter{A}{\pvec a}{}\ \to\ \iInter{A}{\pvec b}{}
\end{equation}
and also an operation $\pvec a \cup \pvec b$ that computes an upper bound of $\pvec a$ and $\pvec b$ w.r.t.\ $\preceq$.
We can then construct $\pvec\psi$ satisfying (\ref{ind-conclusion}) by primitive recursion:
\begin{equation}
\label{eq-primrec-join}
\pvec \psi(m) =
   \left\{
        \begin{array}{ll}
             \pvec t & m = 0  \\
             \pvec \psi(m-1) \cup \pvec \phi(m-1, \pvec \psi(m-1)) & m > 0.
        \end{array}
   \right.
\end{equation}
Now, using the abbreviation 
\begin{equation}
   \beta(m) \equiv \forall n \leq m (\iInter{A(n)}{\pvec \psi(m)}{}),
\end{equation}
%
we see that (\ref{ind-base}) is equivalent to $\beta(0)$, while (\ref{ind-step}) implies (taking $\pvec a = \pvec \psi(m)$ and using (\ref{upper-bound})), 
%
\[ 
   \forall m (\beta(m) \to \beta(m+1)).
\]
Therefore, by induction, $\forall m \beta(m)$, i.e.~(\ref{ind-conclusion}) holds.

In the case of the interpretation by finite sets ($n \triangleleft_{\NN} a \equiv n \in a$ where $a$ ranges over finite sets of natural numbers) the realization of induction is similar, but slightly more involved (see Section~\ref{sec-herbrand}). 

For the precise interpretation ($n \triangleleft_{\NN} a \equiv n = a$) the realization of induction is
similar to (\ref{eq-primrec-join}), but without the join operation:
\begin{equation}
\label{eq-primrec-precise}
\pvec \psi(m) =
   \left\{
        \begin{array}{ll}
             \pvec t & m = 0  \\
             \pvec \phi(m-1, \pvec \psi(m-1)) & m > 0,
        \end{array}
   \right.
\end{equation}
where for $\phi$ to realize the step now means
\begin{equation} \label{ind-step-precise}
    \forall m \forall \pvec a ((\iInter{A(m)}{\pvec a}{}) \to (\iInter{A(m + 1)}{\pvec \phi(m, \pvec a)}{})).
\end{equation}

\section{Instances}
\label{sec-instances-K}

In this section, we consider several concrete choices of base interpretations.

\subsection{Kleene Realizability}
\label{sec-kleene}

Let us start by considering Kleene's original realizability notion, which is based on the partial combinatory algebra $\K$. 
In this case, all realizers are natural numbers and we write the partial application operation, $e(\pvec a)$, as $\{e\}(\pvec a)$
%
%
The target theory is in this case $\HA$, however, in the traditional formulation, i.e.\ without the predicate $\NN$. 


\begin{definition}[Kleene base interpretation] We consider the following base interpretation of the three predicate symbols of $\HA$:
\begin{equation} \label{base-kleene}
\begin{array}{rl}
	\bounded{\,\perp}{\emptys\,}{\emptys}{} & \;\; :\equiv \;\; \perp \\[2mm]
	\bounded{\NN}{n}{m}{} & \;\; :\equiv \;\; n = m \\[2mm]
	\bounded{\;=}{(n, m)}{\emptys}{} & \;\; :\equiv \;\; n = m.
\end{array}
\end{equation}
Let us then write ``$\kleeneInter{A}{\pvec a}{}$'' for the instance of the uniform realizability interpretation ``$\iInter{A}{\pvec a}{}$'' (Def. \ref{def-uniform-realizability}) obtained from the Kleene's first algebra with the base interpretation (\ref{base-kleene}). 
\end{definition}

The following proposition shows that the instantiation above essentially yields Kleene's realizability interpretation\footnote{The (inessential) difference is that Kleene codes tuples of realizers into a single number.}: 

\begin{proposition}[Kleene realizability \cite{Kleene(45)}] \label{prop-kleene} The following equivalences hold:
\begin{align*}
	\kleeneInter{\,(n = m)}{\emptys}{} 
		& \;\; \Leftrightarrow \;
		n = m \\
	\kleeneInter{A \wedge B}{\pvec a, \pvec b}{} 
		& \;\; \Leftrightarrow \;
		(\kleeneInter{A}{\pvec a}{}) \wedge (\kleeneInter{B}{\pvec b}{}) \\
	\kleeneInter{A \to B}{\pvec b}{} 
		& \;\; \Leftrightarrow \;
		\forall \pvec a ((\kleeneInter{A}{\pvec a}{}) \to (\{\pvec b\}(\pvec a) \downarrow 
			\wedge \, \kleeneInter{B}{\{\pvec b\}(\pvec a)}{})) \\
	\kleeneInter{\pexists{n}{\NN} A(n)}{m, \pvec a}{} 
		& \;\; \Leftrightarrow \;
		\kleeneInter{A(m)}{\pvec a}{} \\
	\kleeneInter{\pforall{n}{\NN} A(n)}{\pvec a}{} 
		& \;\; \Leftrightarrow \;
		\forall n (\{\pvec a\}(n) \downarrow \wedge \, \kleeneInter{A(n)}{\{\pvec a\}(n)}{}).
\end{align*}
\end{proposition}
{\bf Proof.} 
Immediate. The equivalences for the qualified quantifiers $\pexists{n}{\NN} A(n)$ and $\pforall{n}{\NN} A(n)$ follow from Remark~(\ref{rem-bounded-quantifers}). \hfill $\Box$

\begin{proposition}[Soundness of Kleene realizability]
\label{prop-kleene-sound}
All theorems of\\ $\HA$ are Kleene-realizable.
\end{proposition}
{\bf Proof.} One can either transfer Kleene's original proof, via a computable coding of tuples, or use the general Soundness Theorem (Thm.~\ref{thm-soundness}) which reduces the problem to showing that the nonlogical axioms of $\HA$ are realizable. The axioms of reflexivity, symmetry, and transitivity for equality do not contain the predicate $\NN$ and are therefore interpreted by themselves.
The realizability of the induction scheme is discussed in Section \ref{sec-nat}. The remaining Peano axioms are either translated into themselves, or realized by the identity function, as one sees from Remark \ref{rem-bounded-quantifers}. \hfill $\Box$

\subsection{Kreisel Modified Realizability}
\label{sec-kreisel}

Instead of using Kleene's first algebra $\K$ we can instead use G\"odel's finite-type primitive recursive functionals, formalised in G\"odel's system $\T$. In this case the realizers are terms of system $\T$, and application $t(s)$ is the usual function application $t(s)$ where $t \colon \rho \to \tau$ and $s \colon \rho$. The crucial difference to Kleene realizability is that all realizers are now total and application is a total operation.
The base interpretation for modified realizability is (\ref{base-kleene}), as for Kleene. 
The target system $\TargetT$ is now Heyting arithmetic in finite types, $\HA^\omega$~ \cite{Troelstra(73)} (or a version of it, depending on the exact choice of the source system $\SourceT$ which may be the theory $\HA$ or a version of $\HA^\omega$).

%
Let us write ``$\kreiselInter{A}{\pvec a}{}$'' for the realizability interpretation ``$\iInter{A}{\pvec a}{}$'' (Def. \ref{def-uniform-realizability}) obtained from the primitive recursive functionals with the base interpretation (\ref{base-kleene}). 

\begin{definition} To each formula $A$ of $\HA$ let us associate a tuple of types $\typeof{A}$ inductively as follows: For the atomic formulas we have
\[
	\typeof{\perp} \; := \; \emptys 
	\qquad
	\typeof{n = m} \; := \; \emptys 
	\qquad
	\typeof{\NN(n)} \; := \; \nat
\]
and inductively
\begin{align*}
	\typeof{A \wedge B} 
		& \; := \; \typeof{A}, \typeof{B} &
	\typeof{A \to B} 
		& \; := \;
		\typeof{A} \to \typeof{B} \\
	\typeof{\exists n A} 
		& \; := \;
		\typeof{A} &
	\typeof{\forall n A} 
		& \; := \;
		\typeof{A}.
\end{align*}
where for tuples of types $\pvec \sigma = \sigma_1, \ldots, \sigma_n$ and $\pvec \rho = \rho_1, \ldots, \rho_m$, we set
\[
\pvec\sigma \to \pvec\rho := \sigma_1\to\ldots\to\sigma_n\to\rho_1,\ldots,\sigma_1\to\ldots\to\sigma_n\to\rho_m.
\]
\end{definition}
Clearly, if $\kreiselInter{A}{\pvec a}{}$, then $\pvec a$ is of type $\typeof{A}$, which we sometimes write as $\kreiselInter{A}{\pvec a^{\typeof{A}}}{}$.

\begin{proposition}[Kreisel modified realizability \cite{Kreisel(59),Kreisel(62)}] \label{prop-kreisel} The following equivalences hold:
\begin{align*}
	\kreiselInter{\,(n = m)}{\emptys}{} 
		& \;\; \Leftrightarrow \;
		n = m \\
	\kreiselInter{A \wedge B}{\pvec a^{\typeof{A}}, \pvec b^{\typeof{B}}}{} 
		& \;\; \Leftrightarrow \;
		(\kreiselInter{A}{\pvec a}{}) \wedge (\kreiselInter{B}{\pvec b}{}) \\
	\kreiselInter{A \to B}{\pvec f^{\typeof{A} \to \typeof{B}}}{} 
		& \;\; \Leftrightarrow \;
		\forall \pvec a^{\typeof{A}} ((\kreiselInter{A}{\pvec a}{}) \to (\kreiselInter{B}{\pvec f(\pvec a)}{})) \\
	\kreiselInter{\exists n^{\NN} A(n)}{m^\nat, \pvec a^{\typeof{A}}}{} 
		& \;\; \Leftrightarrow \;
		\kreiselInter{A(m)}{\pvec a}{} \\
	\kreiselInter{\pforall{n}{\NN} A(n)}{\pvec f^{\nat \to \typeof{A}}}{} 
		& \;\; \Leftrightarrow \;
		\forall n^{\nat} (\kreiselInter{A(n)}{\pvec f(n)}{}).
\end{align*}
\end{proposition}
{\bf Proof.} 
The proof is very similar to that of Proposition \ref{prop-kleene}, except that now 
realizers are
total and therefore definedness statements can be omitted.
%
%
\hfill $\Box$

\begin{remark}[Independence of premise]
One of the reasons why Kreisel introduced modified realizability is that it realizes the following \emph{Independence of Premise} schema
\begin{equation}
\label{eq-ipn}
(A \to \pexists{x}{\NN} B(x)) \to \pexists{x}{\NN }(A \to B(x))
\end{equation}
where $A$ is a negated formula that does not contain $x$ free.
If one drops in (\ref{eq-ipn}) the relativizations to $\NN$, the resulting schema is still modified realizable if one accepts the same schema in the target system. This is so, since
\[\kreiselInter{\,(A \to \exists x B(x)) \to \exists x (A \to B(x))}{\pvec f}{}\] 
is equivalent to 
\[\forall \pvec b\,((\kreiselInter{A}{\emptys}{} \to \exists x (\kreiselInter{B(x)}{\pvec b}{})) \to \exists x (\kreiselInter{A}{\emptys}{} \to \kreiselInter{B(x)}{\pvec f(\pvec b)}{}))\] 
which, when choosing for $\pvec f$ the identity function, follows from another instance of the same schema.
\end{remark}

\subsection{Herbrand Realizability}
\label{sec-herbrand}

For the Herbrand (non-standard) realizability interpretation~\cite{Berg(2012)} we consider an extension of $\HA$ with an extra predicate symbol $\std(\cdot)$. Intuitively, $\NN(n)$ denotes any natural number, standard or non-standard, whereas $\std(n)$ denotes a standard natural number. We then take the total primitive recursive functionals $\T^*$ extended with \emph{star types}, $\tau^*$, denoting non-empty finite sets, and the following base interpretation of the atomic formulas of $\HA$:
\begin{equation} \label{base-herbrand}
\begin{array}{rl}
	\bounded{\,\perp}{\emptys\;}{\emptys}{} & \;\; :\equiv \;\; \perp \\[2mm]
	\bounded{\NN}{n}{\emptys}{} & \;\; :\equiv \;\; {\sf true} \\[2mm]
	\bounded{\std}{n}{S}{} & \;\; :\equiv \;\; n \in S \\[2mm]
	\bounded{=}{(n,m)}{\emptys}{} & \;\; :\equiv \;\; n = m.
\end{array}
\end{equation}
Let us write ``$\herbrandInter{A}{\pvec a}{}$'' for the realizability interpretation ``$\iInter{A}{\pvec a}{}$'' (Def. \ref{def-uniform-realizability}) obtained from $\T^*$ with the base interpretation (\ref{base-herbrand}). 

\begin{proposition}[Herbrand realizability \cite{Berg(2012)} -- variant] \label{prop-Herbrand} The following equivalences hold:
\begin{align*}
	\herbrandInter{\,(n = m)}{\emptys}{} 
		& \;\; \Leftrightarrow \;
		n = m \\
	\herbrandInter{A \wedge B}{\pvec a^{\typeof{A}}, \pvec b^{\typeof{B}}}{} 
		& \;\; \Leftrightarrow \;
		(\herbrandInter{A}{\pvec a}{}) \wedge (\herbrandInter{B}{\pvec b}{}) \\
	\herbrandInter{A \to B}{\pvec f^{\typeof{A} \to \typeof{B}}}{} 
		& \;\; \Leftrightarrow \;
		\forall \pvec a^{\typeof{A}} ((\herbrandInter{A}{\pvec a}{}) \to (\herbrandInter{B}{\pvec f(\pvec a)}{})) \\
	\herbrandInter{\pexists{n}{\NN} A(n)}{\pvec a^{\typeof{A}}}{} 
		& \;\; \Leftrightarrow \;
		\exists n (\herbrandInter{A(n)}{\pvec a}{}) \\
	\herbrandInter{\pforall{n}{\NN} A(n)}{\pvec a^{\typeof{A}}}{} 
		& \;\; \Leftrightarrow \;
		\forall n (\herbrandInter{A(n)}{\pvec a}{}) \\
	\herbrandInter{\pexists{n}{\std} A(n)}{S^{\NN^*}, \pvec a^{\typeof{A}}}{} 
		& \;\; \Leftrightarrow \;
		\exists n \in S (\herbrandInter{A(n)}{\pvec a}{}) \\
	\herbrandInter{\pforall{n}{\std} A(n)}{\pvec f^{\NN^* \to \typeof{A}}}{} 
		& \;\; \Leftrightarrow \;
		\forall S \forall n \in S (\herbrandInter{A(n)}{\pvec f(S)}{}).
\end{align*}
\end{proposition}

Note that the interpretation above differs from the original Herbrand interpretation \cite{Berg(2012)} in two ways. Firstly, our target theory does not have the standard predicate $\std(x)$ as it is eliminated when interpreted as $n \in S$. Consequently, we do not require the quantifications in the interpretation of $A \to B$ and $\pforall{n}{\std} A(n)$ to be over standard objects, as done in \cite{Berg(2012)}.

Secondly, the universal quantification $\pforall{n}{\std} A(n)$ in the Herbrand realizability as defined in~\cite{Berg(2012)} is not interpreted as above, but rather as
\begin{align*}
	\herbrandInter{\pforall{n}{\std} A(n)}{\pvec g^{\NN \to \typeof{A}}}{} 
		& \;\; \Leftrightarrow \;
		\forall n (\herbrandInter{A(n)}{\pvec g (n)}{}).
\end{align*}
However, one can easily go from $\pvec f^{\NN^* \to \typeof{A}}$ satisfying 
\begin{equation}
	\forall S \forall n \in S (\herbrandInter{A(n)}{\pvec f(S)}{})
\end{equation}
to a $\pvec g^{\NN \to \typeof{A}}$ satisfying
\begin{equation}
	\forall n (\herbrandInter{A(n)}{\pvec g (n)}{})
\end{equation}
and vice-versa, namely $\pvec g (n) := \pvec f (\{ n \})$ and $\pvec f (S) := \bigcup_{n \in S} \pvec g (n)$, using the monotonicity property of the Herbrand realizability (see \cite{OX(2020)} for details).

\begin{remark}[Disjunction]
\label{rem-disjunction}
Defining
\begin{equation}
\label{eq-disjunction}
A \vee B\quad :\equiv \quad\exists n\,((n = 0 \to A) \wedge (n \neq 0 \to B))
\end{equation}
one has
\begin{equation}
\label{eq-disjunction-realizer}
\herbrandInter{A \vee B}{\pvec a,\pvec b}{}\quad \Leftrightarrow \quad \herbrandInter{A}{\pvec a}{} \vee \herbrandInter{B}{\pvec b}{}
\end{equation}
which agrees with the nonconstructive interpretation of disjunction in~\cite{Berg(2012)}.
\end{remark}







\subsection{Classical realizability}
\label{sec-classical}

In all the previous examples, the interpretation of falsity ($\bot$) is such that negated formulas $\neg A$, which are abbreviations for $A \to \bot$, do not require any realizer. We can extract computational content from negated formulas, however, by giving some computational meaning to $\bot$, which can be seen as a combination of modified realizability and Friedman and Dragalin $A$-translation \cite{Dragalin(80), Friedman(78)}. For instance, consider the following base interpretation of the atomic formulas of $\HA$, where $P(a^\tau)$ is a new (unary) predicate symbol:
\begin{equation} \label{base-berger} \begin{array}{rl}
	\bounded{\,\perp}{\emptys\,}{a^\tau}{} & \;\; :\equiv \;\; P(a) \\[2mm]
	\bounded{\NN}{n}{m^\NN}{} & \;\; :\equiv \;\; n = m \\[2mm]
	\bounded{\;=}{(n, m)}{a^\tau}{} & \;\; :\equiv \;\; (n = m) \vee P(a).
\end{array}
\end{equation}
If the (classical) realizability is only being applied after a double negation translation, we are in fact in minimal logic (and we no longer need to deal with ex-falso-quodlibet $\bot \to A$), the atomic formulas $n = m$ can be given a simpler interpretation
\begin{equation} 
\begin{array}{rl}
	\bounded{\;=}{(n, m)}{\emptys}{} & \;\; :\equiv \;\; n = m.
\end{array}
\end{equation}

Let us write ``$\bergerInter{A}{\pvec a}$'' for the realizability interpretation ``$\iInter{A}{\pvec a}{}$'' (Def. \ref{def-uniform-realizability}) obtained from the primitive recursive functionals with the base interpretation (\ref{base-berger}). 

\begin{definition} To each formula $A$ of $\HA$ let us associate a tuple of types $\typeof{A}$ inductively as follows: For the atomic formulas we have
\[
	\typeof{\perp} \; := \; \tau
	\qquad
	\typeof{n = m} \; := \; \tau
	\qquad
	\typeof{\NN(n)} \; := \; \NN
\]
and inductively
\begin{align*}
	\typeof{A \wedge B} 
		& \; := \; \typeof{A}, \typeof{B} &
	\typeof{A \to B} 
		& \; := \;
		\typeof{A} \to \typeof{B} \\
	\typeof{\exists n A} 
		& \; := \;
		\typeof{A} &
	\typeof{\forall n A} 
		& \; := \;
		\typeof{A}.
\end{align*}
\end{definition}

\begin{proposition}[Classical modified realizability \cite{BO(02B),BO(02A)}] \label{prop-berger} The following equivalences hold:
\begin{align*}
	\bergerInter{\,\bot}{a} 
		& \;\; \Leftrightarrow \;
		P(a) \\
	\bergerInter{\,(n = m)}{a} 
		& \;\; \Leftrightarrow \;
		(n = m) \vee P(a) \\
	\bergerInter{A \wedge B}{\pvec a^{\typeof{A}}, \pvec b^{\typeof{B}}}
		& \;\; \Leftrightarrow \;
		(\bergerInter{A}{\pvec a}) \wedge (\bergerInter{B}{\pvec b}) \\
	\bergerInter{A \to B}{\pvec f^{\typeof{A} \to \typeof{B}}} 
		& \;\; \Leftrightarrow \;
		\forall \pvec a^{\typeof{A}} ((\bergerInter{A}{\pvec a}) \to (\bergerInter{B}{\pvec f(\pvec a)})) \\
	\bergerInter{\pexists{n}{\NN} A(n)}{m^\NN, \pvec a^{\typeof{A}}} 
		& \;\; \Leftrightarrow \;
		\bergerInter{A(m)}{\pvec a} \\
	\bergerInter{\pforall{n}{\NN} A(n)}{\pvec f^{\NN \to \typeof{A}}} 
		& \;\; \Leftrightarrow \;
		\pforall{n}{\NN} (\bergerInter{A(n)}{\pvec f(n)}).
\end{align*}
\end{proposition}

The main motivation behind giving falsity ($\bot$) computational content is that the negated formula $\neg \neg \exists a^\tau P(a)$ now requires a realizer $\phi$, namely,
\begin{equation}
    \bergerInter{\neg \neg \exists a^\tau P(a)}{\phi} \; \equiv \; 
    \forall f^{\tau \to \tau} (\forall a (P(a) \to P(f(a))) \to P(\phi(f))).
\end{equation}
Therefore, if $\bergerInter{\neg \neg \exists a^\tau P(a)}{\phi}$ then $P(\phi(\lambda a . a))$.

\subsection{Aschieri-Berardi Learning Realizability}
\label{sec-learning}

Assume now that G\"odel's system $\T$ is extended with a new base type ${\mathbf S}$ of ``states''. Let us call the extension $\T[{\mathbf S}]$. For the learning realizability at state $s$, we take the terms of $\T[{\mathbf S}]$ as realizers, and, for some fixed state $s \in {\mathbf S}$, the following base interpretation of the atomic formulas:
\begin{equation} \label{base-berardi}
\begin{array}{rl}
	\bounded{\perp}{\emptys}{\gamma}{s} & \;\; :\equiv \;\; \gamma(s) \neq s \\[2mm]
	\bounded{\NN}{n}{\alpha}{s} & \;\; :\equiv \;\; \alpha(s) = n \\[2mm]
	\bounded{=}{(n, m)}{\gamma}{s} & \;\; :\equiv \;\; (\gamma(s) = s) \to (n = m),
\end{array}
\end{equation}
where $\alpha \colon {\mathbf S} \to \NN$ and $\gamma \colon {\mathbf S} \to {\mathbf S}$. So, falsity ($\perp$) is realized by the state transformer $\gamma$ if $s$ is not a fixed-point of $\gamma$. In other words, it is perfectly fine to reach a contraction if $\gamma$ is learning something from that ($\gamma(s)$ is the improved state of knowledge). Similarly, $n = m$ might be false, as long as $s$ is not a fixed-point for $\gamma$.

Intuitively, the state $s$ will keep track of triples $\langle P, \vec m, n \rangle$, where $n$ is a witness to $\exists n P(\vec m, n)$. We will start with the empty state (nothing is known), and the state transformer $\gamma$, extracted from the proof, updates a given state with new information or returns the same state (no new information is needed), in which case a fixed point is reached. This allows the realizability to interpret any instance of the law-of-excluded middle for $\Sigma_1^0$, i.e. formulas of the kind $\exists n P(\vec m, n)$. 

Let us write ``$\kripkeInter{A}{\pvec a}{s}$'' for the instance of the uniform realizability interpretation ``$\iInter{A}{\pvec a}{}$'' (Def. \ref{def-uniform-realizability}) obtained from $\T[{\mathbf S}]$ and base interpretation (\ref{base-berardi}), for some fixed state $s \in {\mathbf S}$.

\begin{definition} To each formula $A$ of $\HA$ let us associate a tuple of types $\typeof{A}$ inductively as follows: For the atomic formulas we have
\[
	\typeof{\perp} \; := \; {\mathbf S} \to {\mathbf S} 
	\qquad
	\typeof{n = m} \; := \; {\mathbf S} \to {\mathbf S} 
	\qquad
	\typeof{\NN(n)} \; := \; {\mathbf S} \to \NN
\]
and inductively
\begin{align*}
	\typeof{A \wedge B} 
		& \; := \; \typeof{A}, \typeof{B} &
	\typeof{A \to B} 
		& \; := \;
		\typeof{A} \to \typeof{B} \\
	\typeof{\exists n A} 
		& \; := \;
		\typeof{A} &
	\typeof{\forall n A} 
		& \; := \;
		\typeof{A}.
\end{align*}
\end{definition}

\begin{proposition}[Aschieri-Berardi learning realizability] \label{prop-berardi} The following equivalences hold:
\begin{align*}
	\kripkeInter{\,(n = m)}{\gamma^{{\mathbf S} \to {\mathbf S}}}{s} 
		& \;\; \Leftrightarrow \;
		(\gamma(s) = s) \to (n = m) \\[1mm]
	\kripkeInter{A \wedge B}{\pvec a^{\typeof{A}}, \pvec b^{\typeof{B}}}{s} 
		& \;\; \Leftrightarrow \;
		(\kripkeInter{A}{\pvec a}{s}) \wedge (\kripkeInter{B}{\pvec b}{s}) \\[1mm]
	\kripkeInter{A \to B}{\pvec f^{\typeof{A} \to \typeof{B}}}{s} 
		& \;\; \Leftrightarrow \;
		\forall \pvec a^{\typeof{A}} ((\kripkeInter{A}{\pvec a}{s}) \to (\kripkeInter{B}{\pvec f(\pvec a)}{s})) \\[1mm]
	\kripkeInter{\pexists{n}{\NN} A(n)}{\alpha^{{\mathbf S} \to \NN}, \pvec a^{\typeof{A}}}{s} 
		& \;\; \Leftrightarrow \;
		\kripkeInter{A(\alpha(s))}{\pvec a}{} \\[1mm]
	\kripkeInter{\pforall{n}{\NN} A(n)}{\pvec f^{({\mathbf S} \to \NN) \to \typeof{A}}}{s} 
		& \;\; \Leftrightarrow \;
		\forall \alpha^{{\mathbf S} \to \NN} \pforall{n}{\NN} (\alpha(s) = n \to \kripkeInter{A(n)}{\pvec f (\alpha)}{s}).
\end{align*}
\end{proposition}
{\bf Proof}. Let us look at some non-trivial cases. For instance, the axiom of transitivity for equality requires a term $t[\gamma_1, \gamma_2]$ such that
\[
    (\gamma_1(s) = s \to n = i) \wedge (\gamma_2(s) = s \to i = m) \to (t[\gamma_1, \gamma_2](s) = s \to n = m).
\]
Berardi-Aschieri have shown \cite{Aschieri(2010)} that a form of `union' of states $\Cup$ exists so that $t[\gamma_1, \gamma_2](s) := \gamma_1(s) \Cup \gamma_2(s)$ satisfies 
\[
t[\gamma_1, \gamma_2](s) = s \to \gamma_1(s) = s \wedge \gamma_2(s) = s. \] 
Let us look at the case of $\pforall{n}{\NN} A(n)$.
\begin{align*}
	\kripkeInter{\pforall{n}{\NN} A(n)}{\pvec f^{\NN \to \typeof{A}}}{s} 
		& \;\; \equiv \;
			\iInter{\forall n (\NN(n) \to A(n)}{\pvec f}{}) \\
		& \; \stackrel{\textup{D}\ref{def-uniform-realizability}}{\equiv} 
			\forall n (\forall \alpha ((\bounded{\NN}{n}{\alpha}{s}) \to 
				(\pvec f(\alpha) \downarrow \wedge \; \iInter{A(n)}{\pvec f(\alpha)}{}))) \\
		& \;\; \stackrel{(\ref{base-kleene})}{\equiv} \;
			\forall n (\forall \alpha (\NN(n) \wedge (\alpha(s) = n) \to 
				(\iInter{A(n)}{\pvec f(\alpha)}{}))) \\[1mm]
		& \;\;\, \Leftrightarrow \; \;
			\forall \alpha^{{\mathbf S} \to \NN} \pforall{n}{\NN} ((\alpha(s) = n) \to  \kripkeInter{A(n)}{\pvec f(\alpha)}{s}).
\end{align*}
The other cases are straightforward. \hfill $\Box$

\begin{remark} The universal quantification $\pforall{n}{\NN} A(n)$ in the learning realizability is not interpreted as above, but rather as
\begin{align*}
	\kripkeInter{\pforall{n}{\NN} A(n)}{\pvec g^{\NN \to \typeof{A}}}{s} 
		& \;\; \Leftrightarrow \;
		\pforall{n}{\NN} (\kripkeInter{A(n)}{\pvec g (n)}{s}).
\end{align*}
It is easy to see, however, that one can easily go from $\pvec f^{({\mathbf S} \to \NN) \to \typeof{A}}$ satisfying 
\begin{equation}
	\forall \alpha^{{\mathbf S} \to \NN} \pforall{n}{\NN} ((\alpha(s) = n) \to  \kripkeInter{A(n)}{\pvec f(\alpha)}{s})
\end{equation}
to a $\pvec g^{\NN \to \typeof{A}}$ satisfying
\begin{equation}
	\pforall{n}{\NN} (\kripkeInter{A(n)}{\pvec g (n)}{s})
\end{equation}
and vice-versa, namely $\pvec g (n) := \pvec f (\lambda s . n)$ and $\pvec f (\alpha) := \pvec g (\alpha (s))$.
\end{remark}

\section{Conclusion}

We introduced \emph{uniform realizability} as a family of realizability interpretations where computational content is generated exclusively through the interpretation of the basic predicates of a theory while quantifiers are interpreted uniformly. 

A benefit of the uniform presentation is that one can establish a general form of Soundness  (Theorem~\ref{thm-soundness}), so that it is enough to show the realizability of non-logical axioms concerned with the primitive predicates to obtain soundness for a particular instance. 
Another benefit is that it is possible to add arbitrary classically valid formulas as axioms as long as they only contain primitive predicates that are interpreted uniformly since such formulas are interpreted by themselves (Remark~\ref{rem-uniform}).

We discussed five examples of concrete realizability interpretations that can be modelled in that way: Kleene's number realizability, Kreisel's modified realizability, a version of the Herbrand realizability of nonstandard arithmetic, ``classical'' realizability based on the Friedman-Dragalin translation, and the stateful ``learning'' realizability by Aschieri and Berardi. To keep things simple, we considered these interpretations only for first-order source theories $\SourceT$, but most of them could easily be extended to all finite types.

The idea of uniform quantifiers also appears in realizability for second-order Heyting Arithmetic \cite{Troelstra(73)} and in Krivine's classical realizability~\cite{Krivine(2009)}. Krivine's interpretation does not quite fit in our framework though, since it treats the logical constants differently. The idea of uniformity is taken even further in Schwichtenberg's approach~\cite{SchwichtenbergWainer12} where also uniform versions of the propositional connectives are considered. On the type-theoretic side, the idea of uniformity is embodied through intersection types~\cite{BarendregtCoppoDezani83}.

\paragraph{Further work.} 

Some instances of uniform realizability we have discussed, such as Herbrand realizability and ``learning'' realizability, do not exactly coincide with their original formulations. As a result, their soundness must be carefully verified. Another important concern is conservativity: expressing an interpretation as an instance of uniform realizability typically requires a richer language, which includes both qualified and unqualified quantifiers. As noted above, this expanded language allows for the addition of new axioms that are either themselves realizable or classically valid and equivalent to their realizability interpretations, but which cannot be formulated within the original theory. Examples include double negation elimination for uniform formulas and a version of the independence-of-premise schema~(\ref{eq-ipn}) in which qualified existential quantifiers are replaced with unqualified ones. This raises the question of whether the extended theory remains conservative over the original one.

\bibliographystyle{eptcs}

\bibliography{dblogic}

\end{document}